\documentclass[11pt,reqno]{amsart}

\usepackage{amssymb}
\usepackage{amscd}
\usepackage{amsfonts}
\usepackage{mathrsfs}
\usepackage{setspace}
\usepackage{version}
\usepackage{mathtools}
\usepackage{enumerate}
\usepackage[english]{babel}
\usepackage[dvipsnames]{xcolor}

\newtheorem{theorem}{Theorem}[section]

\newtheorem{proposition}[theorem]{Proposition}
\newtheorem{corollary}[theorem]{Corollary}
\theoremstyle{definition}
\newtheorem{definition}[theorem]{Definition}
\newtheorem{example}[theorem]{Example}

\theoremstyle{remark}
\newtheorem{remark}[theorem]{Remark}
\numberwithin{equation}{section}

\newcommand{\FF}{\mathcal{F}}

\newcommand{\MM}{\mathcal{M}}

\newcommand{\PP}{\mathcal{P}}
\newcommand{\QQ}{\mathcal{Q}}

\newcommand{\field}[1]{\mathbb{#1}}

\newcommand{\R}{\field{R}}
\newcommand{\N}{\field{N}}

\newcommand{\Q}{\field{Q}}

\newcommand{\E}{\field{E}}

\renewcommand{\P}{\field{P}}

\newcommand{\ba}{{\rm{ba}}}

\newcommand{\al}{\alpha}

\newcommand{\ep}{\varepsilon}

\renewcommand{\th}{\theta}

\newcommand{\la}{\lambda}

\newcommand{\si}{\sigma}

\newcommand{\om}{\omega}

\newcommand{\Om}{\Omega}

\newcommand{\Bb}{\text{\rm B}_b}
\newcommand{\Max}{\text{\rm{Max}}}
\newcommand{\Min}{\text{\rm{Min}}}

\newcommand{\var}{\text{\rm{var}}}
\newcommand{\VaR}{\text{\rm{V@R}}}
\newcommand{\AVaR}{\text{\rm{AV@R}}}
\newcommand{\Amb}{\text{\rm{Amb}}}

\renewcommand{\ba}{\mathop{\text{\upshape{ba}}}\nolimits}

\newcommand{\mnen}[1]{{#1}}

\newcommand{\erase}[1]{{{#1}}}
\newcommand{\changes}[1]{{{#1}}}

\begin{document}

	\title[Premium principles, risk, and deviation]{A decomposition of general premium principles into risk and deviation}
	
	\author{Max Nendel}
	\address{Center for Mathematical Economics, Bielefeld University}
	\email{max.nendel@uni-bielefeld.de}

	\author{Frank Riedel}
	\address{Center for Mathematical Economics, Bielefeld University}
	\email{frank.riedel@uni-bielefeld.de}		
	
	\author{Maren Diane Schmeck}
	\address{Center for Mathematical Economics, Bielefeld University}
	\email{maren.schmeck@uni-bielefeld.de}

	\thanks{Financial support through the German Research Foundation via CRC 1283 is gratefully acknowledged. The authors thank Hansj\"org Albrecher, Hans-Ulrich Gerber, Marcelo Brutti Righi, Ruodu Wang, and three anonymous referees for their helpful comments and remarks.}

	\date{\today}

	\begin{abstract}
	\noindent  We provide an axiomatic approach to general premium principles in a probability-free setting that allows for Knightian uncertainty. Every premium principle is the sum  of a risk measure, as a generalization of the expected value, and a deviation measure, as a generalization of the variance. One can uniquely identify  a maximal risk measure  and a minimal deviation measure in such decompositions.  We show how previous axiomatizations of premium principles can be embedded into our more general framework. We discuss dual representations of convex  premium principles, and study the consistency of premium principles with a financial market in which insurance contracts are traded. 	
	

	\end{abstract}

	\maketitle
	
	\setcounter{tocdepth}{1}


\section{Introduction}

A premium  principle is a map that assigns a premium  $H(X)$   to a loss   $X$, cf.\ B\"uhlmann \cite{MR580671}, Deprez and Gerber \cite{MR797503}\erase{, Young \cite{doi:10.1002/9780470012505.tap027},} or, for textbook references, Albrecher et al.\ \cite{MR3791478}, Kaas et al.~\cite{kaas2008m}, and Rolski et al.~\cite{MR1680267}. The literature usually assumes that the premium principle is law-invariant in the sense that it depends only on the probability distribution of losses, cf.\ Wang et al.\ \cite{MR1604936}.\ There are instances, however, when the probability distribution is not known exactly and when it might not be easy to estimate, due to a lack of stationarity, missing data, absence of a suitable law of large numbers or a central limit theorem (compare Cairns \cite{Cairns2000} for a discussion in an insurance context). From a more practical point of view,  the International Actuarial Association acknowledges the importance of such uncertainty in Chapter 17 of the risk book \cite{riskbook}: '\textit{Risk} is the effect of variation that results from the random nature of the outcomes being studied (i.e. a quantity susceptible of measurement).~\textit{Uncertainty} involves the degree of confidence in understanding the effect of perils or hazards not easily susceptible to measurement.' Model uncertainty is also widely recognized, for example, in the context of life insurance, cf.~Biagini et al.~\cite{MR3388119}, Bauer et al.\ \cite{bauer2010pricing}, Milevsky et al.\ \cite{MPY}, and Schmeck and Schmidli \cite{schmeck2019mortality}.

In this paper, we thus take a more general position, and model insurance claims as measurable functions, thus allowing for Knightian uncertainty or a ``model-free'' setting. In particular, we  do not assume ex ante that the probability distributions of losses are known to the insurer, nor that premium principles are law-invariant. For a class $C$ of bounded claims\footnote{We point out that we allow for claims taking both positive and negative values, where positive values correspond to losses.}, we impose only two very natural conditions on premium principles. We require that there is no unjustified risk loading, i.e.\ a shift of a loss by a known amount is priced correctly, or 
\begin{equation}\label{eq.p1}
  H(X+m)=H(X)+m \quad \text{for all }X\in C\text{ and }m\in \R, \tag{P1}
\end{equation}
compare Deprez and Gerber \cite{MR797503} and Young \cite{doi:10.1002/9780470012505.tap027}. The textbook Kaas et al.~\cite[Section 5.3.1]{kaas2008m} calls Property \eqref{eq.p1}   a consistency condition\footnote{In the context of monetary risk measures, property (P1) is, up to a sign, usually referred to as \textit{cash additivity}, see e.g.~F\"ollmer and Schied \cite{MR3859905}. }.
Our second  requirement has the form
\begin{equation}\label{eq.p2}
 H(X)\geq H(0)=0 \quad \text{for all }X\in C\text{ with }X\geq 0.\tag{P2}
\end{equation}
Condition \eqref{eq.p2} simply states that an insurer will not be willing to pay money for insuring \emph{pure losses}, i.e.~claims with only positive outcomes.  Since typically insurance claims have only positive outcomes, one could, loosely speaking, interpret \eqref{eq.p2} as a condition stating that premium principles are always nonnegative -- a standard requirement\erase{, see e.g. Young \cite{doi:10.1002/9780470012505.tap027}}. From now on, the term \textit{premium principle} refers to a map that fulfills the Axioms \eqref{eq.p1} and \eqref{eq.p2}.

Our first main result shows that every premium principle can be written as
the sum of  a monetary risk measure $R$ (compare, e.g., F\"ollmer and Schied \cite{MR3859905}) and a deviation measure $D$ (compare Rockafellar and Uryasev \cite{MR3103448})\footnote{\changes{We also refer to Liu et al.~\cite{MR4034687} for an overview on convex risk functionals, a class containing both risk and deviation measures, and to Righi \cite{MR4019241} for a detailed discussion on compositions between risk and deviation measures.}}. Therefore, the simple axioms \eqref{eq.p1} and \eqref{eq.p2} are sufficient to provide a lot of structure to premium principles.  In the classic case, when the probability distribution is known, a typical insurance premium consists of the sum of the fair premium and a multiple of the variance or standard deviation. As the expected loss is a risk measure and the variance a deviation measure, one can think of premium principles as generalizations of this basic approach.

Classic premium principles like the aforementioned variance or standard deviation principle or well-known economic principles can be subsumed under our framework, and we discuss how they can be naturally  generalized  to Knightian uncertainty. We also review more modern notions of quantile-based premia involving Value at Risk or Expected Shortfall, cf.~Rolski et al.~\cite[Section 3.1.3]{MR1680267} and Kaas et al.~\cite[Section 5.6]{kaas2008m}.

It is natural to ask in what sense the risk and the deviation measure can be identified uniquely. In general, this is not the case. However, we show that the premium principle can be uniquely decomposed  into a \emph{maximal} risk measure $R_\Max$ (capturing all risky components of the insurance claim) and a \emph{minimal} deviation measure $D_\Min$ measuring the claim's pure fluctuations. 
The risk measure $R_\Max$ solves a variational problem that is, at least in spirit, akin to the idea of superhedging in finance because  it computes the minimal premium  for a claim $X_0 \in C$ that covers the loss $X$ in every state of the world.  The maximal risk measure  can be  defined on the whole space of claims; the theorem thus provides a natural algorithm to assess the riskiness of arbitrary insurance claims given the available portfolio of contracts that the insurer is able to price.
The minimal deviation, i.e.~the difference between the premium and the maximal risk measure can be seen as a margin for compensating the parts of the claim that cannot be quantified as pure risk\footnote{In a different setup, the existence of a greatest monotone function majorized by a given function $f$ has been discussed  by Kupper and Filipovi\'{c} \cite{MR2324561} and Maccheroni et al.\ \cite{MR2536871}.}.

In a second step, we assume that, in addition to \eqref{eq.p1} and \eqref{eq.p2}, the premium principle $H$ is convex or sublinear. We derive a dual representation of $R_\Max$ in terms of the Fenchel-Legendre transform of $H$. In the sublinear case, we show that there exists a maximal set $\PP$ of probability measures (priors) satisfying
\begin{equation}\label{eq.riskloading}
 H(X)\geq \E_\P(X)\quad \text{for all }X\in C \text{ and }\P\in \PP.
\end{equation}
The latter can be seen as a generalized version of a safety loading, see Castagnoli et al.~\cite{MR1932757}\erase{\ and Young \cite{doi:10.1002/9780470012505.tap027}}. If $\PP=\{\P\}$ consists of a single prior, one ends up with the classical safety loading. In view of equation \eqref{eq.riskloading}, the set $\PP$ can be seen as the set of all priors that are covered by the premium principle in the sense that the premium principle avoids bankruptcy under each model $\P\in \PP$. We refer to $\PP$ as the set of all plausible models. 

In order to cover nonmonotone standard approaches like the variance principle, we do not assume premium principles to be monotone ex ante. However, depending on the context one can, of course, add the condition of monotonicity, or the related no-ripoff condition (cf. Deprez and Gerber \cite{MR797503}, Kaas et al.~\cite[Section 5.3.1]{kaas2008m}, or Young \cite{doi:10.1002/9780470012505.tap027}),
\begin{equation}\label{eq.no-ripoff}
 H(X)\leq \sup X\quad \text{for all }X\in C.
\end{equation}
As we discuss below, monotonicity leads to a vanishing minimal deviation measure. Moreover, we provide a minimal axiomatization of convex monotone premium principles in the spirit of Deprez and Gerber \cite{MR797503} using the no-ripoff condition \eqref{eq.no-ripoff} instead of Axiom \eqref{eq.p2} (see Proposition \ref{lem.sublinearmon}).

 Castagnoli et al.~\cite{MR1932757} introduce an insurance premium of the form   \begin{equation}\label{eq.castagnoli}
 H(X)= \mathbb E_\P(X)+\Amb_\PP(X)
\end{equation}
for a fixed baseline model $\P\in \PP$ and an ambiguity index 
\begin{equation}\label{eq.ambmeas}
 \Amb_\PP(X):=\frac{1}{2}\sup_{\Q,\Q'\in \PP}\mathbb E_\Q(X)-\mathbb E_{\Q'}(X).
\end{equation}
They impose the no-ripoff condition \eqref{eq.no-ripoff}; we show that this condition  implies that the set of priors $\PP$ is dominated by the reference measure $\P$ (see Proposition \ref{prop.monamb}). It is thus important to allow for a more general approach without invoking the no-ripoff condition if one wants to cover undominated sets of priors when considering a premium taking into account an ambiguity measure of the form \eqref{eq.ambmeas}. Such models appear naturally if one models uncertainty about the volatility of diffusions. We refer to Remark \ref{rem.cast} for more details.

Last not least, we discuss extensions of our basic model. Section 4 discusses the important property of law-invariance. We show that the maximal risk measure  inherits law invariance from the premium principle. Moreover, suitably continuous, convex, and law-invariant premia always carry a safety loading. 
We also discuss the relation of the maximal risk measure to superhedging in presence of a competitive market that is used by the insurer to hedge against certain risks using portfolios or securization products that are traded in the market. In the spirit of F\"ollmer and Schied \cite{MR1932379}, we derive equivalent conditions ensuring that the premium principle is consistent with superhedging.
We conclude by showing how our results extend to unbounded claims and more classic setups, where claims are identified via their distributions using the increasing (convex) order instead of the pointwise order.

The paper is structured as follows. In Section \ref{sec.premium}, we introduce the setup and notations, provide the decomposition of a premium principle into risk and deviation, give an explicit description of the maximal risk measure $R_\Max$, and discuss various examples illustrating our notion of a premium principle. Section \ref{sec.dualrep} is devoted to the study of convex and sublinear premium principles. In this context, we discuss dual representations, multiple priors, and baseline models. In Section \ref{sec.4}, we  discuss law invariance, extensions of the model to unbounded claims, and consistency with financial markets. We summarize the main contribution of the paper in Section \ref{sec.conclusion}. The proofs can be found in the Appendix \ref{App.A}.

\section{Premium principles and their decompositions}\label{sec.premium}

\subsection{Model and Notation}
Let  $(\Om,\FF)$ be a measurable space. Denote the space of all bounded, real-valued measurable functions   by $\Bb=\Bb(\Om,\FF)$.  Let $C\subset \Bb$ represent the set of insurance claims covered by a premium policy.~We assume that $0\in C$ and that $X+m\in C$ for all $X\in C$ and $m\in \R$, where, in the notation, we do not differentiate between real constants and constant functions (with real values). Thus, we also consider claims with possibly negative values. We call every measurable function $X\in \Bb$ a \emph{claim}. We use the notation
\[
 \sup X:=\changes{\sup_{\om\in \Om}} X(\om)\quad \text{and}\quad \inf X:=\changes{\inf_{\om\in \Om}} X(\om).
\]
We denote by $\leq$ both the usual order on the reals and the pointwise order on $\Bb$.

\subsection{Premium Principles and a Basic Decomposition}

The central object in our analysis is the following notion of a premium principle.

\begin{definition}\label{def.premium}
 We say that a map $H\colon C\to \R$ is a \textit{premium principle} on $C$ if
 \begin{enumerate}
  \item[(P1)] $H(X+m)=H(X)+m$ for all $X\in C$ and $m\in \R$.
  \item[(P2)] $H(0)=0$ and $H(X)\geq 0$ for all $X\in C$ with $X\geq 0$.
 \end{enumerate}
\end{definition}

Condition (P1) together with $H(0)=0$ implies the common assumption of \textit{no unjustified risk loading}, i.e. $H(m)=m$ for all constant claims $m\in \R$, cf.~Deprez and Gerber \cite{MR797503}\erase{\ and Young \cite{doi:10.1002/9780470012505.tap027}}. Concerning  Property (P2),  note that the condition  $H(0)=0$ is natural for insurance claims.  A typical policy insures \emph{losses} in the sense that the claim is either zero or positive. (P2) ensures that the company or the market take a nonnegative premium for sure losses. It is thus a minimal requirement for a sensible notion of a premium policy.\\

We briefly recall the notion of a (monetary) risk measure and refer to F\"ollmer and Schied \cite{MR3859905} for a detailed discussion on the latter. Since we identify losses with positive claims, we choose a different sign convention than in \cite{MR3859905}, and call a map \changes{$R\colon \Bb\to \R$} a \textit{(monetary) risk measure} if
 \begin{enumerate}
 \item[(R1)] $R(X+m)=R(X)+m$ for all $X\in \Bb$ and $m\in \R$,
 \item[(R2)] $R(0)=0$ and $R(X)\leq R(Y)$ for all $X,Y\in \Bb$ with $X\leq Y$.
 \end{enumerate}
A map $D\colon C\to \R$ is a \textit{deviation measure} (cf.~Rockafellar-Uryasev \cite{MR3103448})  if 
 \begin{enumerate}
  \item[(D1)] $D(X+m)=D(X)$ for all $X\in C$ and $m\in \R$,
  \item[(D2)] $D(0)=0$ and $D(X)\geq 0$ for all $X\in C$.
 \end{enumerate}
 
 Let $R\colon \Bb\to \R$ be a risk measure and $D\colon C\to \R$ be a deviation measure. One easily verifies that the sum 
 \[
  H(X):=R(X)+D(X), \quad \text{for }X\in C,
 \]
 defines a premium principle on $C$. In fact this decomposition into a monetary risk measure and a deviation measure characterizes all premium principles as the following theorem shows. 
 
\begin{theorem}\label{cor.main}
 A map $H\colon C\to \R$ is a premium principle if and only if \begin{equation}\label{eq.decompriskdev} H(X)= R(X)+D(X)\quad \text{for all }X\in C,\end{equation} where \changes{$R\colon \Bb\to \R$} is a risk measure and $D\colon C\to \R$ is a deviation measure. 
\end{theorem}

The theorem shows that premium principles can be decomposed into a \emph{net premium}, given in terms of a risk measure, and a \emph{safety loading} that compensates the insurer for the variability of the damage. In the classic case, when a prior probability distribution $\P$ is given, the typical premium consisting of the sum of the expected loss $ E_\P(X)$ and (a multiple of) the variance of $X$ under $\P$ is a case in point. Note that the expected loss is a risk measure and the variance a deviation measure.\\

It is natural to ask in what sense the risk and the deviation measure can be identified uniquely. In general, there exists a multitude of such decompositions, cf. Section \ref{sec.examples} and Section \ref{sec.dualrep} below. The following theorem shows that the premium principle can be uniquely decomposed into a maximal risk measure $R_\Max$, capturing all risky components of the insurance claim, and a minimal deviation measure $D_\Min$, constituting a part of the premium that cannot be justified by any risk measure.

\begin{theorem}\label{thm.main}\
Let $H\colon C\to \R$ be a premium principle. Define 
\[
  R_\Max(X):=\inf\big\{H(X_0)\, \big|\, X_0\in C,\, X_0\geq X\big\},\quad \text{for }X\in \Bb.
\]
 Then, the map $R_\Max\colon \Bb\to \R$ defines a risk measure, and $R_\Max(X)\leq H(X)$ for all $X\in C$.~Moreover, $D_\Min(X):=H(X)- R_\Max(X)$ defines a deviation measure on $C$, and
 \[
  H(X)=R_\Max(X)+D_\Min(X)\quad \text{for all }X\in C.
 \]
For every other decomposition of the form $H(X)=R(X)+D(X)$, for $X\in C$, with a risk measure $R$ and a deviation measure $D$, we have $R \le R_\Max$ and $D \ge D_\Min$.
\end{theorem}

Theorem \ref{thm.main} shows that one can identify uniquely a \emph{maximal} risk measure and a \emph{minimal} deviation measure whose sum forms the premium principle. 
The risk measure $R_\Max$ solves a variational problem that is, at least in spirit, akin to the idea of superhedging in finance because  it computes the minimal premium that one has to pay for a claim $X_0 \in C$ that covers the loss given by $X$ in every state of the world.  Note that  the maximal risk measure is defined on the whole space of claims $\Bb$; the theorem thus provides a natural algorithm to assess the riskiness of arbitrary insurance claims given the available set of contracts that the insurer is able to price with the premium principle $H$.  

\begin{remark}
We   briefly point out how the decomposition into risk and deviation relates to the monotonicity of premium principles. Notice that, by definition, a premium principle $H$ is \textit{monotone}, in the sense that $H(X)\leq H(Y)$ for all $X,Y\in C$ with $X\leq Y$, if and only if $H$ is a risk measure.\ Therefore, a premium principle $H$ is monotone if and only if $H(X)=R_\Max(X)$ for all $X\in C$, or, equivalently, if $D_\Min(X)=0$ for all $X\in C$. However, this does \textit{not} exclude the existence of a nontrivial decomposition into risk and deviation of the form \eqref{eq.decompriskdev}, i.e.\ there do exist monotone premium principles of the form \eqref{eq.decompriskdev} with $D(X)\neq 0$ for all nonconstant claims $X\in C$, cf.\ Example \ref{ex.meanabs} below. 
\end{remark}

\subsection{Examples}\label{sec.examples}

We illustrate how classic and new approaches of insurance pricing can be subsumed under our framework.

\subsubsection{Classic Premium Principles under a Given Probabilistic Model}

\begin{example}[Ad hoc  premium principles under a given model]
The benchmark premium principle is the \textit{fair premium principle} given by
\[
 H(X)=\E_\P(X), \quad \text{for }X\in \Bb,
\]
where $\P$ is a fixed probability measure on $(\Om,\FF)$. Here, $R_\Max=\E_\P(\,\cdot\,)$ and $D_\Min=0$. In practice, since the fair premium  contains no premium for taking risk, insurers usually add a safety loading, e.g.~in terms of the variance, leading to the \textit{variance principle}
 \[
  H(X)=\E_\P(X)+\frac{\theta}{2} \var_\P(X),\quad \text{for }X\in \Bb,
 \]
 with a constant $\theta \geq 0$.  Here, $R=\E_\P(\, \cdot \,)$ and $D=\frac{\theta}{2}\var_\P(\, \cdot\, )$ is a decomposition of $H$ into risk and deviation. However, as we will see in Example \ref{ex.subldevmeas}, for $\theta>0$, the maximal risk measure $R_\Max$ is given by
 \[
  R_\Max(X)=\max_{\Q\in \PP} \E_\Q(X)-\frac{1}{2\theta} G(\Q|\P),
 \]
 where $\PP$ consists of all probability measures $\Q$ that are absolutely continuous with respect to $\P$ and satisfy 
 \[
  G(\Q|\P):=\var_\P\bigg(\frac{{\rm d}\Q}{{\rm d}\P}\bigg)<\infty.
 \]
 $G$ is the so-called \textit{Gini concentration index}, see e.g.~Maccheroni et al.~\cite{MR2268407},\cite{MR2536871}.
\end{example}

 \begin{example}[Mean absolute deviation]\label{ex.meanabs}
 An example for a monotone premium principle allowing for a nontrivial decomposition \eqref{eq.decompriskdev} with $D\neq 0$ is given by
 \[
  H(X):=\E_\P(X)+\theta \E_\P\big(\big|X-\E_\P(X)\big|\big), \quad \text{for }X\in \Bb,
 \]
 with a fixed probability measure $\P$ on $(\Om,\FF)$ and $0\leq \theta \leq \frac{1}{2}$. In fact, using the inverse triangle inequality, for $X\in \Bb$ with $X\leq 0$, we obtain 
 \begin{align*}
  H(X)&=(1-\theta )\E_\P(X)+\theta\Big( \E_\P\big(\big|X-\E_\P(X)\big|\big)-\E_\P(|X|)\Big)\\
  &\leq (1-\theta )\E_\P(X)+\theta \E_\P(|X|)=(1-2\theta)\E(X)\leq 0.
 \end{align*}
In Proposition \ref{lem.sublinearmon} below, we show that  $H$ is a monotone premium principle and that the maximal decomposition from Theorem \ref{thm.main} is given by $H(X)=R_\Max(X)$ for $X\in \Bb$.  We conclude that the maximal decomposition can be quite different from the intuitive definition of the premium principle that we started with. 
 
 More generally, one can also consider so-called $L^p$-deviation principles, cf.\ Kupper and Filipovi\'{c} \cite[Section 5.1]{MR2324561}. These are given by 
 \[
  H(X):=\E_\P(X)+\theta \E_\P\big(\big|X-\E_\P(X)\big|^p\big)^{1/p},\quad \text{for }X\in \Bb,
 \]
 where $\P$ is again a fixed probability measure on $(\Om,\FF)$, $\theta\geq 0$, and $p\in [1,\infty)$. Here,
 \[
  R_\Max(X)=\max_{\Q\in \PP} \E_\Q(X),\quad \text{for }X\in \Bb,
 \]
 where $\PP$ consists of all probability measures $\Q$ on $(\Om,\FF)$ with $\frac{{\rm d}\Q}{{\rm d}\P}=1+\theta\big(Z-\E(Z)\big)$ for some measurable $Z\in \Bb$ with $\E_\P\big(|Z|^{p/(p-1)}\big)\leq 1$, for $p>1$, and $\sup |Z|\leq 1$, for $p=1$, cf.\ \cite[Proposition 5.1]{MR2324561}.
\end{example}

\begin{example}[Economic premium principles]
 Let $\P$ be a probability measure on $(\Om,\FF)$ and $\ell\colon \R\to \R$ be a nondecreasing and continuously differentiable \emph{loss function}. Considering a random initial endowment $Z\in \Bb$ that could be interpreted as an existing portfolio of insurance contracts, the premium $p:=H(X)$ for an insurance claim $X\in \Bb$ is computed by requiring that the new insurance contract together with the premium $p$ (infinitesimally) does not change the expected loss. This can be expressed by the equation
 \[
  0=\lim_{h\to 0}\frac{\E_\P\big(\ell[Z+h(X-p)]\big)-\E_\P\big(\ell(Z)\big)}{h}=\E_\P\big(\ell'(Z)\cdot (X-p)\big),
 \]
 which results in the so-called \textit{economic premium principle} (see e.g.\ B\"uhlmann \cite{MR580671} and Deprez and Gerber \cite{MR797503})
 \[
  H(X)=\frac{\E_\P\big(X\ell'(Z)\big)}{\E_\P\big(\ell'(Z)\big)},\quad \text{for }X\in \Bb.  
 \]
 Note that this leads again to the mean value principle,  yet under a new measure $\Q$ whose density with respect to the reference probability $\P$ is proportional to the marginal loss $\ell'(Z)$. We can thus write $H(X)=\E_\Q(X)$ and thus, we have $R_\Max(X)=\E_\Q(X)$ for all $X\in \Bb$. For $\P$-a.s.~constant $Z$, the probability measure $\Q$ coincides with $\P$.   
 \end{example}

\subsubsection{Model Uncertainty}

The recent history brought the issue of model uncertainty to center stage; in particular, it has become clear that working under the assumption of a single probability distribution can be too optimistic for insurance companies. New regulations thus ask insurers to take various models into account (stress testing).

\begin{example}[Model uncertainty]
 Instead of a single probability measure $\P$ on $(\Om,\FF)$, we now consider a nonempty set $\PP$ of probability measures on $(\Om,\FF)$. The set $\PP$ can be seen as a set of plausible models, and we thus end up with a setup, where we have model uncertainty w.r.t.\ the models contained in $\PP$. Then, one can consider robust versions of the aforementioned premium principles by regarding worst case scenarios. Examples include:
 \begin{enumerate}[(i)]
  \item A robust variance principle
 \[
  H(X)=\sup_{\P\in \PP}\E_\P(X)+\theta \sup_{\P\in \PP}\var_\P(X),\quad \text{for }X\in \Bb,
 \]
 with $\theta \geq 0$. 
 \item Maxmin expected loss (cf. Gilboa and Schmeidler \cite{MR1000102})
 \[
  H(X):=\sup_{\P\in \PP}\E_\P\big(\ell(X-\inf X)\big)+\inf X,\quad \text{for }X\in \Bb,
 \]
 with a nondecreasing loss function $\ell\colon \R\to \R$.
 \item We now describe a second-order approach to parameter uncertainty in the spirit of  the so-called \textit{smooth ambiguity model}, cf.~Klibanoff et al.~\cite{MR2171327}. Fix a Polish space $\Omega$ endowed with the Borel $\si$-algebra $\FF$, and let  a probability measure $\mu\colon \Sigma\to [0,1]$, the second-order prior, describe the plausibility of a specific probabilistic model $\P\in\PP$, where $\Sigma=\Sigma(\PP)$ denotes the Borel $\si$-algebra on $\PP$ endowed with the vague topology. In the simplest case, where $\ell(x)=\phi(x)=x$, this corresponds to a Bayesian second-order model. This approach can be modified by considering a \changes{continuous nondecreasing} loss function $\ell\colon \R\to \R$ and another \changes{nondecreasing} function $\phi\colon \R\to \R$ that measures the insurer's aversion to model uncertainty. For losses $X\in {\rm C}_b$ with $\inf X=0$, we set 
 \[
  H(X):=\int_\PP \phi\big(\E_\P\big[\ell(X)\big]\big) \, \mu({\rm d}\P),
 \]
 \changes{where ${\rm C}_b$ denotes the space of all continuous and bounded functions $\Om\to \R$.}
 We extend $H$ for arbitrary claims $X\in {\rm C}_b$ with lower bound $m=\inf X$ by setting $H(X)=H(X-m)+m$.
 \end{enumerate}
\end{example}

\begin{example}[Ambiguity indices]
Consider a nonempty set $\PP$ of probability measures on $(\Om,\FF)$ that describes the models that the insurer is willing to consider.  In contrast to the previous example, we now fix a reference model $\P\in \PP$, which can be seen as the (due to some case-dependent reasons) most plausible model.\ The idea is now to consider a safety loading that distinguishes  risk from model uncertainty. The loading for  risk could then, for example, be given by the variance or an $L^p$-deviation measure, whereas the loading for (model) uncertainty is given by 
  \[
   \Amb_\PP(X):=\frac{1}{2}\sup_{\Q,\Q'\in \PP} \E_\Q(X)-\E_{\Q'}(X),\quad \text{for }X\in \Bb.
 \]
 $   \Amb_\PP(X)$ is a measure for the maximal variation of fair premia under the variety of models $\PP$, cf.~Castagnoli et al.~\cite{MR1932757}.
 In Section \ref{sec.dualrep} below, we investigate the case where \begin{equation}\label{eq.examb}
  H(X)=\E_\P(X)+\theta \Amb_\PP(X),\quad \text{for }X\in \Bb,
 \end{equation}
  with $\theta\geq 0$ and $\Amb_\PP$ as a compensation for model uncertainty.
\end{example}

\begin{example}[Quantile-based premium principles]
 Let $\ep \in (0,1)$, $\P$ be a probability measure on $(\Om,\FF)$, and
 \[
  \P_X^{-1}(\la):=\inf\{a\in \R\, |\, \P(X\leq a)\geq \lambda\}, \quad \text{for }X\in \Bb\text{ and }\la\in (0,1).
 \]
 Then, we could consider the \emph{$\ep$-quantile principle}, cf.~Rolski et al.~\cite[Section 3.1.3]{MR1680267} or Kaas et al.~\cite[Section 5.6]{kaas2008m}, 
 \[
  H(X)= \VaR_\P^\ep(-X)=\P_{X}^{-1}(1-\ep), \quad \text{for }X\in \Bb,
 \]
 as a possible premium principle, where $\VaR_\P^\ep$ is also known as the \textit{value at risk} under $\P$ at level $\ep$, cf.~F\"ollmer and Schied \cite{MR3859905}. Here, $R(X)=\VaR_\P^\ep(-X)$ and $D(X)=0$ for all $X\in \Bb$. A major drawback of value at risk is that it is not convex and thus does not reflect diversification effects. Therefore, one often considers the \textit{expected shortfall} or \textit{average value at risk} $\AVaR_\P^\ep$ at level $\ep$, given by
 \[
  \AVaR_\P^\ep(X):=\frac{1}{\ep}\int_0^\ep \VaR_\P^\gamma(X)\, {\rm d}\gamma, \quad \text{for }X\in \Bb.
 \]
 $\AVaR_\P^\ep$ is convex and positive homogeneous, cf. F\"ollmer and Schied \cite{MR3859905}. Alternatively, for $\theta\geq 0$, one can consider the so-called \textit{absolute deviation principle}, cf.~Rolski et al.~\cite[Section 3.1.3]{MR1680267},
 \[
  H(X)=\E_\P(X)+\th \E_\P\left(\left|X-\P_X^{-1}\big(\tfrac12\big)\right|\right),\quad \text{for }X\in \Bb,
 \]
 as a modification of the standard deviation principle. In this case, $R(X)=\E_\P(X)$ and $D(X)=\th \E_\P\big(\big|X-\P_X^{-1}(1/2)\big|\big)$, for $X\in \Bb$, is a decomposition into risk and deviation. Note that
 \[
  D(X)=\frac\theta{2}\left(\AVaR_\P^{\frac12}(-X)+\AVaR_\P^{\frac12}(X)\right)=\theta \Amb_{\QQ_2}(X), \quad \text{for }X\in \Bb,
 \]
 is (up to a constant) an ambiguity index, where $\QQ_2$ consists of all probability measures $\Q\ll\P$ whose density $\frac{{\rm d}\Q}{{\rm d}\P}$ is $\P$-a.s.~bounded by $2$, cf.\ Example \ref{ex.subldevmeas1} below. In it, we also show that, for $\theta\geq 1$, the maximal risk measure $R_\Max$ is given by 
 $$R_\Max(X)=\AVaR_\P^{\frac{1}{1+\th}}(-X),\quad \text{for }X\in \Bb.$$
\end{example}

\begin{example}[Choquet integrals]
Wang et al.\ \cite{MR1604936} derive a representation of premium principles under the assumption of law-invariance, i.e.\ when the probability distribution of losses is known. Consider a  capacity $\gamma$, i.e.\ a set function $\gamma\colon \FF\to [0,1]$ with $\gamma(\emptyset)=0$, $\gamma(\Omega)=1$, and $\gamma(A)\leq \gamma(B)$ for all $A,B\in \FF$ with $A\subset B$. Then, we consider the premium principle given by the \textit{Choquet integral} w.r.t.~$\gamma$
\[
 H(X):=\int_{\inf X}^\infty \gamma\big(\{X> t\}\big)\, {\rm d}t+\inf X,\quad \text{for }X\in \Bb.
\]
Wang et al.\ \cite{MR1604936} show that   when the reference probability $\P$ is fixed, and certain other axioms are satisfied,  every premium principle $H$ can be represented as a Choquet integral w.r.t.~a \textit{distorted probability} $\gamma=g\circ \P$ for   a \textit{distortion function} $g$ (a nondecreasing function on $[0,1]$ with $g(0)=0$ and $g(1)=1$). In this case,
\[
 H(X)=\int_{\inf X}^\infty g\big(\P_X(t)\big)\, {\rm d}t+\inf X,\quad \text{for }X\in \Bb,
\]
where $\P_X(t):=\P(X> t)$ for $t\geq 0$. The premium principle $H$ is monotone, and we obtain $H(X)=R_\Max(X)$ for all $X\in \Bb$.\end{example}

\section{Dual representation of convex premium principles and baseline models}\label{sec.dualrep}

Premium principles should generally reflect the benefits of diversification and the aversion to uncertainty. In this section, we thus consider \emph{convex} premium principles, generalizing the approach of \cite{MR797503} who assume that the probability distribution of claims is known. We  identify the maximal risk measure in the premium's decomposition as a convex risk measure, cf. F\"ollmer and Schied \cite{MR1932379}. Throughout this section, we assume that $C$ is a linear space with $\R\subset C$. We denote the set of all finitely additive probability measures on $(\Om,\FF)$ by $\ba_+^1$. We say that a premium principle $H\colon C\to \R$ is \emph{convex} if
$$H(\la X+(1-\la )Y)\leq \la H(X)+(1-\la)H(Y)\quad\text{for all }\la\in [0,1]\text{ and }X,Y\in C.$$
In this case, we denote the \textit{convex dual} of $H$ by
\[
  H^*(\P):=\sup_{X\in C} \E_\P(X)-H(X)\in [0,\infty],\quad \text{for }\P\in \ba_+^1.
\]
We have the following theorem, which is a partial extension of \cite[Theorem 4.2]{MR2324561}.
\begin{theorem}\label{thm.convex}
 Let $H\colon C\to \R$ a convex premium principle. Then, the maximal risk measure 
 $R_\Max$ in the decomposition of $H$ satisfies 
 \[
  R_\Max(X)=\max_{\P\in \ba_+^1} \E_\P(X)-H^*(\P)\quad \text{for all }X\in \Bb.
 \]
 Moreover,
  \begin{equation}\label{eq.convexmain}
  H^*(\P)=\sup_{X\in \Bb} \E_\P(X)-R_\Max(X)\quad \text{for all }\P\in \ba_+^1.
 \end{equation}
\end{theorem}

By the previous theorem, the convex dual $H^*$ of the premium principle corresponds to the penalty function of its maximal risk measure. $H^*$ thus represents the confidence that the insurer puts on a particular model $\P$ within the class of all possible models. In the sequel, we will refer to
 \[
  \PP:=\big\{\P\in \ba_+^1\, \big|\, H^*(\P)<\infty\big\}
 \]
 as the set of all \emph{plausible models}.\ A priori, the plausible models are only given in terms of finitely additive probability measures.\ However, under additional continuity assumptions on the premium principle $H$, one can ensure that all plausible models are in fact countably additive.
 
 \begin{corollary}\label{cor.contabove}
  Let $H\colon \Bb\to \R$ be a convex premium principle, and assume that $H$ is continuous from above, i.e. $\inf_{n\in \N}H(X_n)= 0$ for all sequences $(X_n)_{n\in \N}\subset \Bb$ with $X_n\leq X_{n+1}$ for all $n\in \N$ and $\inf_{n\in \N}X_n=0$. Then, $R_\Max$ is continuous from above. In particular, all elements of $\PP$ are (countably additive) probability measures.
 \end{corollary}

 \begin{proof}
 Note that, for every sequence $(X_n)_{n\in \N}\subset \Bb$ with $X_n\leq X_{n+1}$ for all $n\in \N$ and $\inf_{n\in \N}X_n=0$,
 \[
  0\leq R_\Max (X_n)\leq H(X_n)\to 0\quad \text{as }n\to \infty.
 \]
 Therefore, $R_\Max$ is continuous from above, and \eqref{eq.convexmain} together with \cite[Theorem 4.22]{MR3859905} implies that all elements of $\PP$ are countably additive.
\end{proof}

If a convex premium principle $H$ is scalable in the sense that it is positively homogeneous, then $R_\Max$ is a coherent risk measure.
\begin{corollary}\label{cor.dualrep}
 Let $H\colon C\to \R$ be a sublinear premium principle, i.e.~$H$ is a convex premium principle, and $H(\la X)=\la H(X)$ for all $X\in C$ and $\la>0$.~Then, the representing maximal risk measure $R_\Max$ is a coherent risk measure, i.e.
 \[
  R_\Max(X)=\max_{\P\in \PP} \E_\P (X)\quad \text{for all }X\in \Bb,
 \]
where the set of plausible models is given by
 \[
  \PP=\big\{\P\in \ba_+^1\, \big|\, \forall X\in C\colon \E_\P (X)\leq H(X)\big\}.
 \]
 \end{corollary}

 \begin{proof}
  This follows directly from Theorem \ref{thm.convex} together with the observation that sublinearity implies $H^*(\P)\in \{0,\infty\}$ for all $\P\in \ba_+^1$.
 \end{proof}

Note that, for all probabilistic models $\P\in \PP$ and all claims $X\in C$,
\begin{equation}\label{eq.safetyload}
 H(X)\geq \E_\P( X )
\end{equation}
if and only if $H^*(\P)=0$. In particular, a sublinear premium principle $H$ incorporates a so-called \textit{safety loading} under each plausible model $\P\in \PP$, cf.~Castagnoli et al.~\cite{MR1932757}\erase{, Young \cite{doi:10.1002/9780470012505.tap027},} and Deprez and Gerber \cite{MR797503}. In the next step, we analyze in more detail the minimal deviation measure of the premium's decomposition. For $\P\in \PP$, we define
\begin{equation}\label{eq.decomplin0}
  D_\P(X):=H\big(X-\E_\P( X)\big),\quad \text{for }X\in C.
\end{equation}
By \eqref{eq.safetyload}, $D_\P$ defines a deviation measure if and only if $H^*(\P)=0$. In this case, we have 
\begin{equation}\label{eq.decomplin}
 H(X)=\E_\P (X)+D_\P(X)\quad \text{for all }X\in C,
\end{equation}
and the deviation measure $D_\P$ can be seen as the profit for accepting the \mnen{aleatoric} risk of $X$ under the model $\P$. Moreover, Equation \eqref{eq.decomplin} provides a model-dependent decomposition of the premium principle $H$ into a risk measure and a deviation measure. We have the following relation between the minimal deviation measure $D_\Min$ and the family $(D_\P)_{\P\in \PP}$.

\begin{corollary}\label{cor.dualdev}
 Let $H\colon C\to \R$ be a convex premium principle. Then,
 \[
  D_\Min(X)=\min_{\P\in \PP} D_\P(X) +H^*(\P) \quad \text{for all }X\in C.
 \]
\end{corollary}

\begin{proof}
By Theorem \ref{thm.convex},
\[
 D_\Min(X)=H(X)-R_\Max(X)=\min_{\P\in \PP} H\big(X-\E_\P(X)\big)+H^*(\P)\quad \text{for all }X\in C.
\]
 The statement now follows from Equation \eqref{eq.decomplin0}.
\end{proof}

\begin{example}\label{ex.subldevmeas}
 Let $\P$ be a probability measure, $\theta> 0$, and consider
 \[
  H(X):=\E_\P(X)+\frac{\theta}{2}\var_\P(X)\quad\text{for all }X\in \Bb.  
 \]
 Let $\Q\in \PP$. Then, Corollary \ref{cor.contabove} implies that $\Q$ is countably additive and absolutely continuous w.r.t.~$\P$, where the latter follows from the inequality $\E_\Q(X)\leq H(X)$ for all $X\in \Bb$. Let $Z:=\frac{{\rm d}\Q}{{\rm d}\P}$. We can write  \begin{align*}
  H^*(\Q)&=\inf_{X\in \Bb} \E_\Q(X)+H(X)=\inf_{X\in \Bb} \E_\P\left[X\bigg(1-Z+\frac{\theta}{2}\big(X-\E_\P(X)\big)\bigg)\right] 
 \end{align*}
With the help of Cauchy-Schwarz inequality, one can then show that 
$$   H^*(\Q)=  \inf\bigg\{\al\E_\P\left(X^2\right) \, \bigg|\, \al\leq 0, \; \al X=1-Z+\frac{\theta}{2}\big(X-\E_\P(X)\big)\bigg\}, $$
compare the Appendix of Maccheroni et al.\ \cite{MR2536871} for more details.
 Note that the equality $\al X=1-Z+\frac{\theta}{2}\big(X-\E_\P(X)\big)$ implies that $\E_\P(X)=0$, which, in turn, implies that $X=\big(\al-\frac\theta 2\big)^{-1}(1-Z)$. Hence,
 \[
  H^*(\Q)=\inf_{\al\leq 0}\al\bigg(\al-\frac\theta 2\bigg)^{-2}\var_\P\bigg(\frac{{\rm d}\Q}{{\rm d}\P}\bigg).
 \]
 Since $\frac{{\rm d}}{{\rm d}\alpha}\al\big(\al-\frac\theta 2\big)^{-2}=0$ if and only if $\al=-\frac\theta 2$, we obtain that
 \[
  H^*(\Q)=-\frac{1}{2\theta} \var_\P\bigg(\frac{{\rm d}\Q}{{\rm d}\P}\bigg)
 \]
 is (up to the factor $-\frac{1}{2\theta}$) the \textit{Gini concentration index}. By Theorem \ref{thm.convex},
 \[
  R_\Max(X)=\max_{\Q\in \PP} \E_\Q(X)- \frac{1}{2\theta}\var_\P\bigg(\frac{{\rm d}\Q}{{\rm d}\P}\bigg)\quad \text{for all }X\in \Bb.
 \]
 \end{example}
 
 \begin{example}\label{ex.subldevmeas1}
  Let $\P$ be a probability measure, $\theta\geq 0$, and consider
 \[
  H(X)=\E_\P(X)+ \th \E_\P\left(\left|X-\P_X^{-1}\big(\tfrac12\big)\right|\right),\quad \text{for }X\in \Bb.
 \]
 Then, by \cite[Lemma 4.46]{MR3859905},
 \begin{align*}
  \E_\P\left(\left|X-\P_X^{-1}\big(\tfrac12\big)\right|\right)&=\E_\P\left(\left(X-\P_X^{-1}\big(\tfrac12\big)\right)^-\right)+\E_\P\left(\left(X-\P_X^{-1}\big(\tfrac12\big)\right)^+\right)\\
  &=\frac{1}{2}\big(\AVaR_\P^{\frac{1}{2}}(-X)+\AVaR_\P^{\frac{1}{2}}(X)\big)
 \end{align*}
 Recall that, for $\ep\in (0,1)$,
 \[
  \AVaR_\P^\ep(X)=\max_{\Q\in \QQ_{1/\ep}}\E_\Q(-X)\quad \text{for all }X\in \Bb,
 \]
 where, for $\al\geq 1$, $\QQ_\al$ denotes the set of all probability measures $\Q\ll \P$ whose density is $\P$-a.s.~bounded by $\al$, cf. \cite[Theorem 4.47]{MR3859905}. Therefore, the set $\PP$ related to $\R_\Max$ is given by the set of all probability measures $\Q^*$ of the form
 \begin{equation}\label{eq.ex3.5}
  \Q^*=\P+\frac{\th}{2} \big(\Q-\Q'\big)
 \end{equation}
 with $\Q,\Q'\in \QQ_2$. We show that $\PP$ consists of all probability measures $\Q^*\ll\P$ with
 \begin{equation}\label{eq.ex3.5-1}
  1-\theta\leq \frac{{\rm d}\Q^*}{{\rm d}\P}\leq 1+\theta\quad  \P\text{-a.s.} 
 \end{equation}
  In particular, for $\theta\geq 1$, $\PP=\QQ_{1+\theta}$, which implies that $$R_\Max(X)=\AVaR_\P^{\frac{1}{1+\th}}(-X)\quad \text{for all }X\in \Bb.$$
 In fact, by the previous argumentation, it follows that every $\Q^*\in \PP$ is of the form \eqref{eq.ex3.5}, which, in turn, implies that it satisfies \eqref{eq.ex3.5-1}. Now, assume that $\Q^*\ll\P$ is a probability measure, which satisfies \eqref{eq.ex3.5-1}. For $\theta=0$, it follows that $\Q^*=\P\in \PP$. Hence, assume that $\theta>0$, and define
 \[
  Z:=\frac{2}{\theta}\left(\frac{{\rm d}\Q^*}{{\rm d}\P}-1\right).
 \]
 Then, $|Z|\leq 2$ $\P$-a.s., $\E_\P(Z)=0$, and, by H\"older's inequality, $\frac{\E_\P(|Z|)}{2}\leq 1$. Define
 \[
  Y:=Z^+ +1-\frac{|Z|}{2}\quad \text{and}\quad Y':=Z^-+ 1-\frac{|Z|}{2}.
 \]
 Then, $0\leq Y\leq 2$ and $0\leq Y'\leq 2$ $\P$-a.s. Moreover, $Y-Y'=Z$ $\P$-a.s. and
 \[
  \E_\P(Y)=\E_\P(Z^+)+1-\frac{\E_\P(|Z|)}{2}=1=\E_\P(Z^-)+1-\frac{\E_\P(|Z|)}{2}=\E_\P(Y').
 \]
 Hence, Equation \eqref{eq.ex3.5} is satisfied with $\Q:=Y{\rm d}\P$ and $\Q':=Y'{\rm d}\P$.
 \end{example}

We say that a premium principle $H\colon C\to \R$ is \textit{monotone} if $H(X)\leq H(Y)$ for all claims $X,Y\in C$ with $X\leq Y$.\ Throughout the remainder of this section, we discuss the relation to monotone sublinear premium principles that Castagnoli et al.~consider in \cite{MR1932757}. More precisely, we show that, in the convex case, replacing Axiom (P2) in the definition of a premium principle by a so-called \emph{internality condition}, cf. \cite{MR1932757}, implies the monotonicity of the premium principle, and thus, together with positive homogeneity, leads to the objects considered in \cite{MR1932757}.~A similar result can be found in Deprez and Gerber \cite[Theorem 3]{MR797503}.

\begin{proposition}\label{lem.sublinearmon}
 Let $H\colon C\to \R$ a convex map that satisfies (P1). Then, the following statements are equivalent: 
  \begin{enumerate}
 \item[(i)] $H$ is a monotone premium principle,
 \item[(ii)] $H(0)=0$ and $H(X)\leq 0$ for all $X\in C$ with $X\leq 0$, i.e. $H$ is internal.
 \end{enumerate}
\end{proposition}
Note that (ii) together with (P1) implies the standard no-ripoff condition \eqref{eq.no-ripoff}.

The following proposition is a partial extension of Theorem 3 in Castagnoli et al.~\cite{MR1932757}, where statement (i) is a reformulation of Axiom P.7 in \cite{MR1932757}.
\begin{proposition}\label{prop.monamb}
 Let $H\colon C\to \R$ be a sublinear premium principle, and define
 \[
  \Amb_\PP(X):=\frac{1}{2}\big(R_\Max(X)+R_\Max(-X)\big)=\frac{1}{2}\max_{\Q,\Q'\in \PP} \E_\Q(X)-\E_{\Q'} (X),
 \]
 for $X\in \Bb$. Then, for every $\P\in \PP$, the following statements are equivalent:
  \begin{enumerate}
  \item[(i)] $\E_\P( X)=\frac{1}{2}\big(R_\Max(X)-R_\Max(-X)\big)$ for all $X\in \Bb$,
  \item[(ii)] $\PP$ is symmetric with center $\P$, i.e. $2\P-\Q\in \PP$ for all $\Q\in \PP$,
  \item[(iii)] $\Amb_\PP(X)=\max_{\Q\in \PP}\E_\Q (X)-\E_\P(X)$ for all $X\in \Bb$,
   \item[(iv)] $D_\P(X)=D_\Min(X)+\Amb_\PP(X)$ for all $X\in C$.
 \end{enumerate}
  In this case, $R_\Max$ is dominated by $\P$, i.e.~every $\Q\in \PP$ is absolutely continuous w.r.t.~$\P$, and $\P$ is countably additive if and only if every $\Q\in \PP$ is countably additive.
\end{proposition}

\begin{remark}\label{rem.cast}
  Let us discuss  the implications of Proposition \ref{prop.monamb} in relation to Castagnoli et al.\ \cite{MR1932757} who  consider sublinear premium principles $H\colon \Bb\to \R$ of the form
  \begin{equation}\label{eq.ambcast}
   H(X):=\E_\P(X)+\Amb_\PP(X) 
  \end{equation}
  with the additional requirement that $H$ is internal. Proposition \ref{lem.sublinearmon}  implies that $H$ is monotone, leading  to the equation
  \[
   \Amb_\PP(X)=R_\Max(X)-\E_\P(X)=\max_{\Q\in \PP} \E_\Q (X)-\E_\P(X).
  \]
  Therefore,   in the setup chosen in \cite{MR1932757}, all elements of $\PP$ are absolutely continuous w.r.t.~the baseline model (center) $\P$; in other words, undominated sets of probability measures are implicitly excluded. However,  undominated sets of plausible models appear quite naturally, for example, when considering a Brownian Motion with Knightian uncertainty  about  the volatility parameter, see e.g.~Peng~\cite{MR2397805},\cite{MR2474349} or Soner et al.~\cite{MR2746175},\cite{MR2842089}. Hence, replacing the internality axiom P.1 in \cite{MR1932757} by the apparently similar assumption (P2) in Definition \ref{def.premium} has a huge impact, allowing also for premium principles of the form \eqref{eq.ambcast} with nonsymmetric and thus undominated sets $\PP$.  \end{remark}

In the following example, we describe a basic setup that leads to a nonsymmetric set of priors $\PP$.

\begin{example}
 Consider the setup \eqref{eq.ambcast} from \cite{MR1932757} with $\Om=\N$, endowed with the $\si$-algebra $\FF=2^\N$ (power set). For $n\in \N$, we consider the measure
 \[
  \P_n:=\frac{1}{n}\sum_{k=1}^n\delta_k,
 \]
 where $\delta_k$ denotes the Dirac measure with barycenter $k\in \N$. We then consider the monotone premium principle
 \[
  H(X):=\sup_{n\in \N} \E_{\P_n}(X)=\sup_{n\in \N}\frac{1}{n}\sum_{k=1}^n X(k),\quad \text{for }X\in \Bb.
 \]
 One readily verifies that the set $\PP$ consists only of probability measures $\P$ of the form
 \begin{equation}\label{eq.excastagnoli}
  \P=\sum_{n\in \N}\la_n\delta_n
 \end{equation}
 with a nonincreasing sequence $(\la_n)_{n\in \N}\subset [0,1]$ summing up to $1$. Assume that there existed some $\P\in \PP$ with $2\P-\P_n\in \PP$ for all $n\in \N$. Then, $\P$ is of the form \eqref{eq.excastagnoli} with a nonincreasing sequence $(\la_n)_{n\in \N}\subset [0,1]$. On the other hand, $2\P-\P_n\in \PP$ for all $n\in \N$, which, in particular, means that
 \begin{equation}\label{eq.excastagnoli1}
  2\lambda_n-\frac1n\geq 2\la_{n+1}\quad \text{for all }n\in \N.
 \end{equation}
 However, Equation \eqref{eq.excastagnoli1} implies that
 \[
  \la_{n+1}=\la_1+\sum_{k=1}^n\big(\la_{k+1}-\la_k\big)\leq \la_1-\sum_{k=1}^n\frac1{2k}\to -\infty,\quad \text{as }n\to \infty,
 \]
 leading to a contradiction.\ By means of Proposition \ref{prop.monamb}, we may conclude that there exists no $\P\in \PP$ with
 \[
  H(X)= \E_\P (X)+ \Amb_\PP(X)\quad \text{for all }X\in \Bb.
 \]
 \mnen{That is, the right-hand side of the previous equation does not define a premium principle in the sense of Castagnoli et al.~\cite{MR1932757}, whereas it defines a premium principle in the sense of Definition \ref{def.premium}.}
\end{example}

\section{Additional properties and extensions of the model}\label{sec.4}

\subsection{Law-invariance of premium principles}

In this section, we focus on the special case, where the premium principle $H$ is law-invariant with respect to a fixed reference probability measure $\P$. 

For $X\in \Bb$,  let $\P_X(z):=\P(X\leq z)$ denote the distribution function of $X$ depending on $z \in \mathbb R$. We say that a premium principle $H\colon \Bb\to \R$ is \textit{law-invariant} if $$H(X)=H(Y)\quad \text{ for all }X,Y\in \Bb\text{ with }\P_X=\P_Y.$$ Law invariance implies  that $H$ is \textit{dominated} by $\P$, i.e. $H(X)=H(Y)$ for all $X,Y\in \Bb$ with $X=Y$ $\P$-a.s. If $H$ is dominated by $\P$, then $R_\Max$ is dominated by $\P$ as the following lemma shows.

\begin{proposition}\label{prop.dominance}
  Let $\P$ be a probability measure on $(\Om,\FF)$ and $H\colon \Bb\to \R$ be a premium principle that is dominated by $\P$. Then, $R_\Max$ is dominated by $\P$.
\end{proposition}

\begin{remark}
Let $H\colon \Bb\to \R$ be a convex premium principle, which is dominated by a probability measure $\P$ on $(\Om,\FF)$. Then, the previous lemma together with Theorem \ref{thm.convex} and \cite[Lemma 4.32]{MR3859905} implies that, for $\Q\in \ba_+^1$, $H^*(\Q)<\infty$ implies that $\Q\ll \P$, i.e. $\Q(N)=0$ for all $N\in \FF$ with $\P(N)=0$. Let $\ba_+^1(\P)$ denote the set of all $\Q\in \ba_+^1$ with $\Q\ll \P$. Then, the set $\PP$ of all plausible models is a subset of $\ba_+^1(\P)$ and
\[
 R_\Max(X)=\max_{\Q\in \ba_+^1(\P)}\E_\Q(X)-H^*(\Q)\quad \text{for all }X\in \Bb.
\]
\end{remark}

We now establish the even stronger property   of law-invariance for $R_\Max$ on sufficiently rich probability spaces. 

\begin{proposition}\label{prop.lawinv}
Let $\P$ be a probability measure on $(\Om,\FF)$ and assume that $(\Om,\FF,\P)$ is an atomless. Let  $H\colon \Bb\to \R$ be a law-invariant premium principle.   Then, $R_\Max$ is law-invariant.
\end{proposition}

Law-invariant convex risk measures on atomless probability spaces  that are continuous from above  enjoy a very particular structure, see e.g.\ \cite[Theorem 4.62]{MR3859905}. In particular, they  always admit a safety loading w.r.t.\ the reference probability measure $\P$. This leads to the following result.

\begin{corollary}
    Let  $(\Om,\FF,\P)$ be an atomless probability space and $H\colon \Bb\to \R$ be a convex and law-invariant premium principle that  is continuous from above.\ Then, $H(X)\geq \E_\P(X)$ for all $X\in \Bb$, 
\end{corollary}

\begin{proof}
 This is a direct consequence of Proposition \ref{prop.lawinv}, Corollary \ref{cor.contabove}, and \cite[Corollary 4.65]{MR3859905}.
\end{proof}

\subsection{Superhedging and market consistency}\label{sec.hedging}

The integration of insurance and finance has been a central issue of research in the last years, cf. Schweizer \cite{MR1817231} and the references therein. In this section, we consider premium principles that are consistent with a given financial market (or liquidly traded insurance contracts). We will identify the maximal risk measure in the premium's decomposition as the so-called superhedging risk measure.\\

The financial market is modeled by a linear subspace $M \subset C$, where $C$ is again assumed to be a linear space, and a nonnegative linear price functional $F \colon M \to \mathbb R$. Assuming $M$ to be a linear space and $F\colon M\to \R$ to be linear corresponds to a competitive market without frictions. \mnen{We would like to point out that our model can also be used for markets with frictions. That is, the linearity of the price functional $F$ can be replaced by sublinearity, and $M$ can be assumed to be a convex cone instead of a linear subspace. In this case, $F$ would resemble the ask price for securization products that are traded in the market or, in other words, the price an insurer has to pay for ``selling'' the risk of a claim to the market.} Nonnegativity is a no arbitrage condition as it requires 
\begin{equation}\label{eq.consistency}
F(X)\ge 0
\end{equation} for nonnegative claims $X \ge 0$. Without loss of generality, we assume that $F(1)=1$, i.e. the interest rate that is implicit in $F$ is zero. We call
$$\MM=\big\{ \P\in \ba_+^1 \,\big|\, \forall X_0\in M\colon \E_\P(X_0)=F(X_0)\big\}$$
the set of martingale measures for the financial market.\\

Throughout this section, we consider a sublinear premium principle. We assume that the premia charged by our insurer coincide with market prices on $M$, i.e.~$H|_M=F$. The condition $H|_M=F$ expresses the fact that the insurer cannot charge a premium above market prices due to competition. We introduce the set $$M_0:=\big\{X_0\in M\, |\, F(X_0)=0\big\}$$ of all claims that are traded on the market with price $0$. In the sequel, we consider the superhedging risk measure
\[
 R_*(X):=\inf\big\{m\in \R\, \big|\, \exists X_0\in M_0\colon m+X_0\geq X\big\}\quad \text{for all }X\in \Bb.
\]
The superhedging risk measure amounts to the cost of staying on the safe side with the help of the products that are already being traded liquidly in the market. 
Note that $R_*$ is well-defined, since $M_0$ is nonempty. Moreover, $R_\Max\leq R_*$ since $H|_M=F$.

\begin{proposition}\label{prop.hedging}
Let $H$ be sublinear. Then, the following statements are equivalent:
\begin{enumerate}
 \item[(i)] The maximal risk measure in the decomposition of $H$  is the superhedging functional $R_*$, i.e. $R_\Max=R_*$.
 \item[(ii)] The premium principle $H$ is based on the use of securization products, i.e. for all $X\in C$, there exists some $X_0\in M$ with $X_0\geq X$ and $H(X) \geq F(X_0)$.
 \item[(iii)] The plausible models for $H$ coincide with the martingale measures, i.e. $\PP=\MM$.
\end{enumerate}
\end{proposition}

\subsection{Premium principles for unbounded random variables}

In this section, we discuss how our notion of premium principles and, in particular, the decomposition into risk and deviation can be extended to spaces of unbounded random variables and more general structures. Instead of $\Bb$, we consider a general vector space $L$ of claims endowed with a preorder $\leq$. Typical examples for $L$ include the following spaces:
\begin{itemize}
 \item Orlicz spaces, Orlicz hearts, and, as special cases, $L^p$-spaces for a fixed dominating reference measure.
 \item Robust Orlicz spaces (cf.\ \cite{lienen} and the references therein) or closures of lattices of bounded continuous functions w.r.t.\ a robust $L^p$-norm (see e.g.\ \cite{Peng}) in the absence of a dominating reference measure.
 \item Spaces of (say bounded) random variables together with the increasing (convex) order.
\end{itemize}
We consider a nonempty subset $C\subset L$ of claims and a fixed nonempty subspace $M\subset C$ of claims, representing cash, i.e.\ highly liquid claims. In the setup of Section \ref{sec.premium}, $M$ corresponds to the set of constant claims. Alternatively, as in the setup of Section \ref{sec.hedging}, $M$ could also correspond to a market of liquidly traded securization products. We assume that $C$ is \textit{exhaustive} in the sense that, for all $X\in L$, there exists some $X_0\in C$ with $X_0\geq X$. Note that, in the setup of Section \ref{sec.premium}, i.e.\ for $L=\Bb$, a sufficient condition for $C$ to be exhaustive is that $C$ contains all constant claims. We further assume that $X+m\in C$ for all $X\in C$ and $m\in M$ and that
\begin{equation}\label{eq.ass.order}
 X+m\leq Y+m\quad \text{for all }m\in M\text{ and }X,Y\in L \text{ with }X\leq Y.
\end{equation}
 This condition is automatically satisfied in preordered vector spaces.
We say that a map $H\colon C\to \R$ is a \textit{premium principle} if
\begin{itemize}
 \item[(P1')] $H(X+m)=H(X)+H(m)$ for all $X\in C$ and $m\in M$.
 \item[(P2')] $H(0)=0$ and $H(X)\geq 0$ for all $X\in C$ with $X\geq 0$.
 \item[(P3')] $-\infty<\inf\big\{H(X_0)\, \big|\, X_0\in C,\, X_0\geq X\big\}$ for all $X\in L$.
\end{itemize}
We would like to briefly discuss the new axioms (P1') - (P3') for a premium principle in view of Definition \ref{def.premium}. (P1') and (P2') are the analogous versions of (P1) and (P2). The new axiom (P3') is a nondegeneracy condition, ensuring that, by increasing the claim size of a fixed claim $X$, the premium cannot become arbitrarily small. In the setup of Section \ref{sec.premium}, i.e.\ for $L=\Bb$, (P3') is automatically satisfied since, by (P2), $H(X_0)\geq \inf X$ for all $X\in \Bb$ and $X_0\geq X$. We say that a map $R\colon L\to \R$ is a \textit{risk measure} if
\begin{itemize}
 \item[(R1')] $R(X+m)=R(X)+R(m)$ for all $X\in C$ and $m\in M$.
 \item[(R2')] $R(0)=0$ and $R(X)\leq R(Y)$ for all $X,Y\in L$ with $X\leq Y$.
\end{itemize}
A map $D\colon C\to \R$ is called a \textit{deviation measure} if
\begin{itemize}
 \item[(D1')] $D(X+m)=D(X)$ for all $X\in C$ and $m\in M$.
 \item[(D2')] $D(0)=0$ and $D(X)\geq 0$ for all $X\in C$.
\end{itemize}
We obtain the following extension of Theorem \ref{thm.main}.

\begin{theorem}\label{thm.main1}
 Let $H\colon C\to \R$ be a premium principle, and define
 \[
  R_\Max(X):=\inf\big\{ H(X_0)\, \big|\, X_0\in C,\, X_0\geq X\big\}\quad \text{for all }X\in L.
 \]
 Then, $R_\Max\colon L\to \R$ is a risk measure with $R_\Max(X)\leq H(X)$ for all $X\in C$ and $R_\Max(m)=H(m)$ for all $m\in M$. In particular, $D_\Min(X):=H(X)-R_\Max(X)$, for $X\in C$, defines a deviation measure, and
 \[
  H(X)=R_\Max(X)+D_\Min(X).
 \]
 Let $R\colon L\to \R$ be a risk measure and $D\colon C\to \R$ be a deviation measure with
 \[
  H(X)=R(X)+D(X).
 \]
 Then, $R(X)\leq R_\Max(X)$ for all $X\in L$ and $D(X)\geq D_\Min(X)$ for all $X\in C$.
\end{theorem}

We now specialize again on the convex case, assuming that $C$ is a linear space. As in Section \ref{sec.dualrep}, we say that a map $H\colon C\to \R$ is convex if
\[
 H\big(\la X+(1-\la )Y\big)\leq \la H(X)+(1-\lambda) H(Y)\quad \text{for all }X,Y\in C\text{ and }\la\in [0,1].
\]
We start with the following observation.

\begin{remark}\label{rem.convcash}
 Let $C$ be a linear space and $H\colon C\to \R$ be a convex premium principle. Then, $H|_M$ is linear. In fact, $H(-m)=H(m)$ for all $m\in M$ since
 \[
  0=H(0)=H(m-m)=H(m)+H(-m)\quad \text{for all }m\in M.
 \]
 Again, since $H(0)=0$, $H(km)=kH(m)$ for all $k\in \N$ and $m\in M$. By convexity of $H$, it follows that $H(\la m)\leq \la H(m)$ for all $\la\in \R$ and $m\in M$, which immediately yields that $H(\lambda m)=\la H(m)$ for all $\lambda\in \R$ and $m\in M$. Since $H(m+n)=H(m)+H(n)$ for all $m,n\in M$, it follows that $H|_M$ is linear.
\end{remark}

We denote by $L'_+$ the set of all monotone linear functionals $\mu\colon L\to R$. Recall that, by the Namioka-Klee theorem, every monotone linear functional on a Banach lattice is continuous. Therefore, if $L$ is a Banach lattice, $L_+'$ is a subset of the topological dual space $L'$ of $L$. For $L=\Bb$, $L'_+$ is the set of all finitely additive measures $\mu$ $\mu(\Omega)<\infty$. For $L=L^p$ with $p\in [1,\infty)$, it holds $L'_+=L^q_+$, where $q\in (1,\infty]$ is the conjugate exponent.\ The following extension of Theorem \ref{thm.convex} forms the basis for an analogous discussion as in Section \ref{sec.dualrep} in a generalized setup.

\begin{theorem}\label{thm.convex1}
Let $\leq$ be \textit{vector preorder}, i.e. a preorder consistent with addition and scalar multiplication, $C$ be a linear space, and $H\colon C\to \R$ be a convex premium principle. For $\mu\in L'_+$, let
\[
 H^*(\mu):=\sup_{X\in C} \mu (X)-H(X)\in [0,\infty].
\]
Then,
\[
 R_\Max(X)=\max_{\mu\in L'_+} \mu (X)-H^*(\mu)\quad \text{for all }X\in \Bb.
\]
Moreover,
 \begin{equation}\label{eq.convexmain1}
  H^*(\mu)=\sup_{X\in L} \mu (X)-R_\Max(X)\quad \text{for all }\mu\in L'_+,
 \end{equation}
 and $H^*(\mu)<\infty$ implies that $\mu (m)=H(m)$ for all $m\in M$.
\end{theorem}

\section{Conclusion}\label{sec.conclusion}
We clarify the structure of premium principles in a very general setting that allows for model uncertainty. Every premium principle is the sum of a monetary risk measure and a deviation measure. We also relate the maximal risk measure in such a decomposition to hedging practices in finance, and thus provide a link between insurance and financial mathematics. 

We see the main value of our approach in providing a conceptual framework to think about premium principles even in situations of model uncertainty. As such, we provide a theoretical justification for different approaches that are already being used, and we provide a guideline for   new approaches that might come in the future when the practices of finance and insurance will grow ever closer together.

\begin{appendix}
\section{Proofs}\label{App.A}

For the sake of clarity, we prove Theorem \ref{thm.main} before Theorem \ref{cor.main}.

\begin{proof}[Proof of Theorem \ref{thm.main}]
 First, note that $R_\Max\colon \Bb\to \R$ is well-defined since $\sup X\in \Bb$ with $\sup X\geq X$ and $H(X_0)\geq H\big(\inf X\big)$ for all $X_0\in C$ with $X_0\geq X$. By definition, $R_\Max(X)\leq H(X)$ for all $X\in C$. Hence, $D_\Min(X)=H(X)- R_\Max(X)\geq 0$ for all $X\in C$. Moreover, $H(X_0)\geq H(0)=0$ for all $X_0\in C$ with $X_0\geq 0$, which implies that $R_\Max(0)=0$. In particular, $D_\Min(0)=H(0)- R_\Max(0)=0$. We will now show that $R_\Max$ defines a risk measure. First, observe that $R_\Max(X)\leq H(Y_0)$ for $X,Y\in \Bb$ with $X\leq Y$ and $Y_0\in C$ with $Y_0\geq Y$. Taking the infimum over all $Y_0\in C$ with $Y_0\geq Y$, it follows that $R_\Max(X)\leq R_\Max(Y)$. Now, let $X\in \Bb$, $m\in \R$ and $X_0\in C$ with $X_0\geq X$. Then, 
 \[
  R_\Max(X+m)\leq H(X_0+m)= H(X_0)+m.
 \]
 Taking the infimum over all $X_0\in C$ with $X_0\geq X$ implies that $R_\Max(X+m)\leq R_\Max(X)+m$. On the other hand,
 \[
  R_\Max(X)+m= R_\Max(X+m-m)+m\leq R_\Max(X+m).
 \]
 This also shows that, for $X\in C$ and $m\in \R$, $$D_\Min(X+m)=H(X+m)- R_\Max(X+m)=H(X)-R_\Max(X)=D_\Min(X).$$
 Let $R\colon \Bb\to \R$ be a risk measure with $ R(X)\leq H(X)$ for all $X\in C$. Then, for all $X\in \Bb$ and $X_0\in C$ with $X_0\geq X$,
 \[
  R(X)\leq R(X_0)\leq H(X_0).
 \]
 Taking the infimum over all $X_0\in C$ with $X_0\geq X$, we may conclude that $R(X)\leq R_\Max(X)$ for all $X\in \Bb$.
\end{proof}

\begin{proof}[Proof of Theorem \ref{cor.main}]
 This is a direct consequence of Theorem \ref{thm.main}, choosing $R=\R_\Max$ and $D=D_\Min$.
\end{proof}

\begin{proof}[Proof of Theorem \ref{thm.convex}]
 We first show that $R_\Max\colon \Bb\to \R$ is convex. Let $X,Y\in \Bb$ and $\la\in [0,1]$. Then, for $X_0,Y_0\in C$ with $X_0\geq X$ and $Y_0\geq Y$,
 \[
 R_\Max\big(\la X+(1-\la)Y\big)\leq H \big(\la X_0+(1-\la)Y_0\big)\leq \la H(X_0)+(1-\la)H(Y_0).
 \]
 Taking the infimum over all $X_0,Y_0\in C$ with $X_0\geq X$ and $Y_0\geq Y$, we obtain that $R_\Max$ is convex. Since $R_\Max$ is a convex risk measure, it follows that, see e.g.~F\"ollmer and Schied \cite[Theorem 4.12]{MR3859905},
 \[
  R_\Max(X)=\max_{\P\in \ba_+^1} \E_\P(X)-R_\Max^*(\P)\quad \text{for all }X\in \Bb,
 \]
 where $R_\Max^*(\P):=\sup_{X\in \Bb} \E_\P(X)-R_\Max(X)$ for $\P\in \ba_+^1$. It remains to show \eqref{eq.convexmain}, i.e. $H^*(\P)=R_\Max^*(\P)$ for all $\P\in \ba_+^1$. Since $R_\Max(X)\leq H(X)$ for all $X\in C$, it follows that
 \[
  R_\Max^*(\P)\geq \sup_{X\in C} \E_\P (X)-R_\Max(X)\geq H^*(\P) \quad\text{for all }\P\in \ba_+^1.  
 \]
 In particular, there exists some $\P\in \ba_+^1$ with $H^*(\P)<\infty$. Therefore,
 \[
  R(X):=\sup_{\P\in \ba_+^1} \E_\P(X)-H^*(\P),\quad \text{for }X\in \Bb,
 \]
 defines a risk measure. Since $R_\Max^*(\P)\geq H^*(\P)$ for all $\P\in \ba_+^1$, it follows that
 \[
  R_\Max(X)\leq R(X)\leq H(X)\quad \text{for all }X\in C.
 \]
 By the maximality of $R_\Max$, we may conclude that $R_\Max=R$. In particular,
 \[
  H^*(\P)\geq \E_\P( X)-R(X)=\E_\P(X)-R_\Max (X) \quad \text{for all }X\in \Bb\text{ and }\P\in \ba_+^1.
 \]
 By definition of $R_\Max^*$, it follows that $H^*(\P)\geq R_\Max^*(\P)$ for all $\P\in \ba_+^1$. 
 \end{proof}

 \begin{proof}[Proof of Proposition \ref{lem.sublinearmon}]
 Trivially, (i) implies (ii). We first show that (ii) implies (P2), let $X\in C$ with $X\geq 0$. Then, by Condition (ii),
 \[
  0\leq-H(-X)\leq H(X),
 \]
 where the second inequality follows from $0=2H(0)\leq H(X)+H(-X)$. In order to prove the monotonicity, first notice that, due to (P1) and (ii),
 \begin{equation}\label{eq.proofmon}
  H(X)=H(X-\sup X)+\sup X\leq \sup X\quad \text{for all }X\in C.
 \end{equation}
 Now, let $X,Y\in C$ with $X\leq Y$. Then, by \eqref{eq.proofmon}, for all $\lambda\in (0,1)$,
 \begin{align*}
  H(X)&\leq \la H(Y)+(1-\la)H\bigg(\frac{X-\la Y}{1-\la}\bigg)\leq \la H(Y)+\sup (X-\la Y)\\
  &\leq \la H(Y)+(1-\la)\sup X.
 \end{align*}
 Letting $\lambda \to 1$, we obtain that $H(X)\leq H(Y)$. 
 \end{proof}

 \begin{proof}[Proof of Proposition \ref{prop.monamb}]
 For all $X\in \Bb$, 
 \[
  \Amb_\PP(X)= R_\Max (X)-\frac{1}{2}\big(R_\Max(X)-R_\Max(-X)\big)
 \]
 and
 \[
  \max_{\Q\in \PP} \E_\Q( X)-\E_\P (X)=R_\Max(X)-\E_\P(X).
 \]
 Therefore, $\Amb_\PP(X)=\max_{\Q\in \PP} \E_\Q( X)-\E_\P(X)$ for all $X\in \Bb$ if and only if $\E_\P(X)=\frac{1}{2}\big(R_\Max(X)-R_\Max(-X)\big)$ for all $X\in \Bb$. On the other hand, if $\E_\P(X)=\frac{1}{2}\big(R_\Max(X)-R_\Max(-X)\big)$ for all $X\in \Bb$, then, for all $X\in \Bb$ and $\Q\in \PP$,
 \[
  2\E_\P(X)-\E_\Q(X)\leq 2\E_\P(X)+\R_\Max(-X)= R_\Max(X),
 \]
 i.e. $2\P-\Q\in \PP$. Next, assume that $2\P -\Q\in \PP$ for all $\Q\in \PP$. Then, for all $X\in \Bb$,
 \begin{align*}
  \frac{1}{2}\big(R_\Max(X)-R_\Max(-X)\big)&=\frac{1}{2}\big(\max_{\Q\in \PP} \E_\Q( X)+\min_{\Q'\in \PP} \E_{\Q'}( X)\big)\\
  &\leq \frac{1}{2}\max_{\Q\in \PP}\big(\E_\Q( X)+(2\E_\P( X)-\E_\Q( X)\big)=\E_\P( X).
 \end{align*}
Using a symmetry argument, this implies that $\frac{1}{2}\big(R_\Max(X)-R_\Max(-X)\big)=\E_\P( X)$ for all $X\in \Bb$. We have thus established the equivalence of (i) - (iii). In order to prove the remaining equivalence, first observe that, for all $\P\in \PP$,
 \begin{align*}
  \E_\P(X)+D_\P(X)&=H(X)=R_\Max (X)+D_\Min(X)\\
  &=\frac{1}{2}\big(R_\Max(X)-R_\Max(-X)\big)+D_\Min(X)+\Amb_\PP(X)
 \end{align*}
 The equivalence between (i) and (iv) in now an immediate consequence of the previous equation. Under (ii), it follows that $\E_\Q (X)\leq 2\E_\P( X)$ for all $X\in \Bb$ with $X\geq 0$ and all $\Q \in \PP$. Choosing $X=1_N$ for $N\in \FF$ with $\P (N)=0$, it follows that every $\Q\in \PP$ is absolutely continuous w.r.t.~$\P$. On the other hand, let $\Q\in \PP$ and $(X_n)_{n\in \N}\subset \Bb$ with $X_{n+1}\leq X_n$ for all $n\in \N$ and $\inf_{n\in \N}X_n=0$. If $\P$ is countably additive, then
 \[
  0\leq \E_\Q( X_n)\leq 2\E_\P(X_n)\to 0\quad \text{as }n\to \infty,
 \]
 which shows that $\Q$ is countably additive.
\end{proof}

 \begin{proof}[Proof of Proposition \ref{prop.dominance}]
 Let $X,Y\in \Bb$ with $X=Y$ $\P$-a.s. Let $Y_0\in \Bb$ with $Y_0\geq Y$. Then, $X_0:=Y_01_{\{X=Y\}}+X1_{\{X\neq Y\}} \in \Bb$ with $X_0\geq X$ and $X_0=Y_0$ $\P$-a.s. Hence,
 \[
  R_\Max(X)\leq H(X_0)=H(Y_0).
 \]
 Taking the infimum over all $Y_0\in \Bb$ with $Y_0\geq Y$, we obtain that $R_\Max(X)\leq R_\Max(Y)$. A symmetry argument yields that $R_\Max(X)=R_\Max(Y)$.
\end{proof}

\begin{proof}[Proof of Proposition \ref{prop.lawinv}]
For $X\in \Bb$ and $y\in (0,1)$, let $\P_X^{-1}(y):=\inf \{x\in \R\, |\, \P_X(x)> y\}$ denote the right-continuous inverse of $\P_X$. Note that, for $X\in \Bb$, $\P_X^{-1}\colon (0,1)\to \R$ is bounded and $\P_X^{-1}(U)\in \Bb$ has the same distribution as $X$ for every uniformly distributed random variable $U\in \Bb$. Let $X,Y\in \Bb$ with $\P_X=\P_Y$. Moreover, let $Y_0\in \Bb$ with $Y_0\geq Y$. Then, $\P_X^{-1}\leq\P_Y^{-1}\leq \P_{Y_0}^{-1}$. By \cite[Lemma A.32]{MR3859905}, there exists some uniformly distributed random variable $U\in \Bb$ with $X=\P_X^{-1}(U)$. Then,
\[
 R_\Max(X)\leq R_\Max\big(\P_{Y_0}^{-1}(U)\big)\leq H\big(\P_{Y_0}^{-1}(U)\big)=H(Y_0).
\]
Taking the infimum over all $Y_0\in \Bb$ with $Y_0\geq Y$, it follows that $R_\Max(X)\leq R_\Max(Y)$. By a symmetry argument, we may conclude that $R_\Max(X)=R_\Max(Y)$.
\end{proof}

\begin{proof}[Proof of Proposition \ref{prop.hedging}]
 First observe that, since $H|_M=F$
 \[
  R_*(X)=\inf\big\{H(X_0)\, \big|\,  X_0\in M,\, X_0\geq X\big\}\quad \text{for all }X\in \Bb.
 \]
 Therefore, by Theorem \ref{thm.convex}, statement (i) is equivalent to statement (iii). It remains to show the equivalence of (i) and (ii). In order to establish this equivalence, first assume that $R_\Max=R_*$. Then, $R_*(X)\leq H(X)$ for all $X\in C$. By definition of $R_*$, this means that, for all $X\in C$, there exists some $X_0\in M$ with $X_0\geq X$ and $F(X_0)\leq H(X)$. Now, assume that, for all $X\in C$, there exists some $X_0\in M$ with $X_0\geq X$ and $F(X_0)\leq H(X)$. Then, by definition of $R_*$, it follows that $R_*(X)\leq H(X)$ for all $X\in C$. The maximality of $R_\Max$ together with $R_\Max\leq R_*$, which is due to $H|_M=F$, thus implies that $R_\Max=R_*$.
\end{proof}
  
 \begin{proof}[Proof of Theorem \ref{thm.main1}]
 The fact that $C$ is exhaustive together with (P3') ensures that $R_\Max\colon L\to \R$ is well-defined. By definition, $R_\Max(X)\leq H(X)$ for all $X\in C$. Hence, $D_\Min(X)=H(X)- R_\Max(X)\geq 0$ for all $X\in C$. Moreover, by (P1'), (P2'), and \eqref{eq.ass.order}, $$H(X_0)=H(X_0-m)+H(m)\geq H(m)\quad\text{for all }m\in M\text{ and }X_0\in C\text{ with }X_0\geq m,$$ which implies that $R_\Max(m)=H(m)$ for all $m\in M$. Since $0\in M$, $R_\Max(0)=H(0)=0$ and $D_\Min(0)=H(0)- R_\Max(0)=0$. Next, we show that $R_\Max$ defines a risk measure. First, observe that $R_\Max(X)\leq H(Y_0)$ for $X,Y\in L$ with $X\leq Y$ and $Y_0\in C$ with $Y_0\geq Y$. Taking the infimum over all $Y_0\in C$ with $Y_0\geq Y$, it follows that $R_\Max(X)\leq R_\Max(Y)$ for all $X,Y\in L$ with $X\leq Y$. Now, let $X\in L$, $m\in M$ and $X_0\in C$ with $X_0\geq X$. Then, by \eqref{eq.ass.order}, it follows that $X_0+m\geq X+m$, which, in turn, implies that
 \[
  R_\Max(X+m)\leq H(X_0+m)= H(X_0)+H(m)=H(X_0)+H(m).
 \]
 Taking the infimum over all $X_0\in C$ with $X_0\geq X$ implies that $R_\Max(X+m)\leq R_\Max(X)+H(m)$. On the other hand,
 \begin{align*}
  R_\Max(X)+H(m)&= R_\Max(X+m-m)+H(m)\leq R_\Max(X+m)+H(-m)+H(m)\\
  &=R_\Max(X+m),
 \end{align*}
 where the last equality follows from the fact that $0=H(0)=H(m)+H(-m)$, i.e. $H(-m)=-H(m)$, for all $m\in M$. Since $H(m)=R_\Max(m)$ for all $m\in M$, it follows that $R_\Max$ is cash additive. This also shows that, for $X\in C$ and $m\in M$, $$D_\Min(X+m)=H(X+m)- R_\Max(X+m)=H(X)-R_\Max(X)=D_\Min(X).$$
 Let $R\colon \Bb\to \R$ be a risk measure and $D\colon C\to \R$ be a deviation measure with $H(X)=R(X)+D(X)$ for all $X\in C$. Then, for all $X\in L$ and $X_0\in C$ with $X_0\geq X$,
 \[
  R(X)\leq R(X_0)=H(X_0)-D(X_0)\leq H(X_0).
 \]
 Taking the infimum over all $X_0\in C$ with $X_0\geq X$, we may conclude that $R(X)\leq R_\Max(X)$ for all $X\in L$ and $D_\Min(X)=H(X)-R_\Max(X)\leq H(X)-R(X)=D(X)$ for all $X\in C$.
\end{proof}

\begin{proof}[Proof of Theorem \ref{thm.convex1}]
 We first show that $R_\Max\colon L\to \R$ is convex. Let $X,Y\in L$ and $\la\in [0,1]$. Then, for $X_0,Y_0\in C$ with $X_0\geq X$ and $Y_0\geq Y$,
 \[
 R_\Max\big(\la X+(1-\la)Y\big)\leq H \big(\la X_0+(1-\la)Y_0\big)\leq \la H(X_0)+(1-\la)H(Y_0).
 \]
 Taking the infimum over all $X_0,Y_0\in C$ with $X_0\geq X$ and $Y_0\geq Y$, we obtain that $R_\Max$ is convex. Let $R_\Max^*(\mu):=\sup_{X\in L} \mu (X)-R_\Max(X)$ for every linear functional $\mu\colon L\to \R$. The Hahn-Banach theorem implies that
 \[
  R_\Max(X)=\max_{\mu\in \PP} \mu (X)-R_\Max^*(\mu)\quad \text{for all }X\in L,
 \]
 where $\PP$ is the set of all linear functionals $\mu\colon L\to R$ with $R_\Max^*(\mu)<\infty$. Since, by Remark \ref{rem.convcash}, $R_\Max|_M=H|_M$ is linear and $\mu(m)\leq R_\Max(m)+R_\Max^*(\mu)$ for all $m\in M$, it follows that $\mu |_M=R_\Max|_M=H|_M$ for all $\mu\in \PP$. Let $\mu\in \PP$ and $X,Y\in L$ with $X\leq Y$. Then,
 \[
  \mu (X)-\mu (Y)=\la \mu \big(\tfrac{1}{\lambda}(X-Y)\big)\leq \la \Big(R_\Max\big(\tfrac{1}{\lambda}(X-Y)\big)+R_\Max^*(\mu)\Big)\leq \la R_\Max^*(\mu)
 \]
 for all $\la\in (0,\infty)$. Letting $\la\to 0$, it follows that $\mu (X)\leq \mu Y$, i.e. $\mu$ is monotone. In particular, 
  \[
  R_\Max(X)=\max_{\mu\in L'_+} \mu (X)-R_\Max^*(\mu)\quad \text{for all }X\in L.
 \]
 It remains to show \eqref{eq.convexmain1}, i.e. $H^*(\mu)=R_\Max^*(\mu)$ for all $\mu\in L'_+$. Since $R_\Max(X)\leq H(X)$ for all $X\in C$, it follows that
 \[
  R_\Max^*(\mu)\geq \sup_{X\in C} \mu (X)-R_\Max(X)\geq H^*(\mu) \quad\text{for all }\mu\in L'_+.
 \]
 In particular, there exists some $\mu\in L'_+$ with $H^*(\mu)<\infty$. Therefore,
 \[
  R(X):=\sup_{\mu\in L'_+} \mu (X)-H^*(\mu),\quad \text{for }X\in L,
 \]
 defines a risk measure. Since $R_\Max^*(\mu)\geq H^*(\mu)$ for all $\P\in L'_+$, it follows that
 \[
  R_\Max(X)\leq R(X)\leq H(X)\quad \text{for all }X\in C.
 \]
 By the maximality of $R_\Max$, we may conclude that $R_\Max=R$. In particular,
 \[
  H^*(\mu)\geq \mu (X)-R(X)=\mu (X)-R_\Max (X) \quad \text{for all }X\in L\text{ and }\mu\in L'_+.
 \]
 By definition of $R_\Max^*$, it follows that $H^*(\mu)\geq R_\Max^*(\mu)$ for all $\mu\in L'_+$. 
 \end{proof}

\end{appendix}

\end{document}